\documentclass[onecolumn, 12pt]{IEEEtran}  
\usepackage[utf8]{inputenc}
\usepackage{algorithm}
\usepackage{multirow}
\usepackage[dvipsnames]{xcolor}
\usepackage{hhline}

\usepackage{bbm}
\usepackage{amsfonts, amsmath, amsthm, amssymb, cite,algorithmic}
\usepackage{tikz, multirow}

\usepackage{soul}
\usepackage{pgfplots}
\usepackage{mathtools}
\usepackage{dsfont}
\usepackage{flushend}
\flushend
\pgfplotsset{compat=newest}

\definecolor{darkgray}{RGB}{64,64,64}
\definecolor{litegray}{RGB}{192,192,192}
\tikzstyle{block}=[draw, rectangle, minimum height=1cm, text width=2cm, text centered, draw=darkgray, font=\small]
\tikzstyle{block_medium}=[draw, rectangle, minimum height=1.5cm, text width=2cm, text centered, draw=darkgray, font=\small]
\tikzstyle{block_large}=[draw, rectangle, minimum height=2.5cm, text width=2cm, text centered, draw=darkgray, font=\small]
\tikzstyle{block_small}=[draw, rectangle, minimum height=1cm, text width=1cm, text centered, draw=darkgray, font=\small]
\tikzstyle{line} = [draw, -latex]
\usepackage{mathtools}
\newcommand{\defeq}{\vcentcolon=}

\newtheorem{theorem}{Theorem}

\newtheorem{lemma}[theorem]{Lemma}

\newcommand{\g}{\boldsymbol{g}}
\renewcommand{\gg}{\overline{\g}}
\newcommand{\q}{\boldsymbol{q}}
\newcommand{\A}{\mathcal{A}}
\newcommand{\X}{\mathcal{X}}
\newcommand{\Y}{\mathcal{Y}}
\renewcommand{\u}{\boldsymbol{u}}
\renewcommand{\v}{\boldsymbol{v}}

\renewcommand{\L}{\mathcal{L}}
\newcommand{\x}{\boldsymbol{x}}
\renewcommand{\S}{\boldsymbol{S}}
\newcommand{\T}{\boldsymbol{T}}
\newcommand{\y}{\boldsymbol{y}}
\newcommand{\0}{\boldsymbol{0}}

\newcommand{\s}{\mathrm{supp}}
\newcommand{\err}{\mathrm{erfc}}

\newcommand{\w}{\mathrm{wt}}
 
\renewcommand{\d}{\mathrm{d}}
\newtheoremstyle{agdTheorem}{\parskip}{\parskip}{\itshape}{\parindent}{\bfseries}{}{0pt}{\thmname{#1}\thmnumber{~#2}.\thmnote{~\textnormal{#3.}}\quad}
\theoremstyle{agdTheorem}
\newtheorem{rmrk}{Remark}
\newtheorem{exmp}{Example}
\newtheoremstyle{agdDefinition}{\parskip}{\parskip}{}{\parindent}{\bfseries}{}{0pt}{\thmname{#1}\thmnumber{~#2}.\thmnote{~\textnormal{#3.}}\quad}
\theoremstyle{agdDefinition}

\title{\LARGE \bf
	Weight Distributions for  Successive Cancellation Decoding of Polar Codes
}

\author{Rina Polyanskaya, Mars Davletshin, and
	Nikita Polyanskii
		\thanks{Rina Polyanskaya is with the Institute for Information Transmission Problems (email: rev-rina@yandex.ru).}
		\thanks{Mars Davletshin is with the Moscow Research Center, Huawei \mbox{Technologies} Co., Ltd  (email: davletshin.mars1@huawei.com).}
	\thanks{Nikita Polyanskii is with the Technical University of Munich. The research was conducted in part during May - October 2017 with the Moscow Research Center, Huawei Technologies Co., Ltd
		(email: nikitapolyansky@gmail.com). N. Polyanskii was funded in part by the German Research Foundation (Deutsche Forschungsgemeinschaft, DFG) under Grant No. WA3907/1-1. }%
}

\begin{document}

	\maketitle
	\thispagestyle{empty}
	\pagestyle{empty}

	\begin{abstract}
		
		In this paper, we derive the exact weight distributions that emerge during each stage of successive cancellation decoding of polar codes. Though we do not compute the distance spectrum of polar codes, the results allow us to get an estimate of the decoding error probability and to show a link between the first nonzero components of the weight distribution and the partial order between the synthetic channels.  Also, we establish the minimal distance between two cosets associated with two paths that differ in two positions. This can be regarded as a first step toward analyzing the weight distributions for  successive cancellation list decoding. 
		
	\end{abstract}
	
		\begin{IEEEkeywords}
		Polar codes, weight distribution,
		closest coset decoding, partial order, successive cancellation decoding.	
	\end{IEEEkeywords}
	\section{Introduction}
	Polar codes, introduced by Ar{\i}kan~\cite{arikan2009channel}, provably achieve the symmetric capacity of any binary-input memoryless  symmetric channels (B-MSC) with encoding and decoding complexity $\Theta(N \log N)$, where $N$ is the block length of the code.   Within the 5G standardization process, polar codes have been adopted as channel codes for uplink and downlink control information of the enhanced mobile broadband communication service.
	
	Multilevel codes are based on partitioning and, thus, multistage decoding is the most natural one to be performed~\cite{pottie1989multilevel}. Polar codes with  successive cancellation (SC) decoding can be represented in this way~\cite{arikan2009channel,trifonov2012efficient}. A typical multilevel code construction employs small codes to get a larger one. Polar codes are obtained by taking a Kronecker power of a square kernel matrix and expurgating some rows using a specific criterion. It is known~\cite[Section X]{arikan2009channel},\cite{arikan2008performance} that polar codes with good distance properties turn out to have a poor performance under SC decoding.
	To evaluate the error rate provided by a multistage decoder, it is natural to calculate the weight distribution (WD) between cosets (or the distance spectrum of component codes) at all the stages~\cite{herzberg1997spectrum}. However, only the minimal distance for  SC decoding of polar codes is known at present~\cite[Lemma~6.2]{korada2009phd}. The aim of our paper is to calculate the WDs at all the stages of  SC decoding.
	
	We want to point out that we do not compute the distance spectrum of polar codes in this paper. Attempts pursuing the latter were undertaken by many authors. For example, the authors of~\cite{liu2014distance} proposed to send the all-zero codeword over a noisy channel and decode the received vector utilizing  the successive cancellation list decoder~\cite{tal2015list}. After decoding, the weights of codewords from the list are calculated and a special weight function is updated. The paper~\cite{valipour2013probabilistic} suggests a way how to compute a probabilistic weight distribution expression efficiently. The authors of~\cite{zhang2017enhanced} put forward an idea how to get an approximate distance spectrum of polar codes with large length using the  spectrum of short ones and some probabilistic assumption on appearing ones in codewords.	
	\subsection{Outline}
	The rest of the paper is organized as follows. In Section~\ref{defSect}, we give key definitions and notations of polar codes and the WDs associated with SC decoding. We derive the WDs  and focus our attention on their first nonzero component in Section~\ref{wd for sc}. To obtain an algorithm computing the WDs, we exploit a similar idea as in~\cite{fossorier1997weight}, where a $|u|u+v|$ construction is investigated. Also, we find a natural connection between the first nonzero component of the WDs and the partial order~\cite{schurch2016partial,bardet2016algebraic}.  The minimal distance between cosets for SC list decoding is discussed in Section~\ref{wd for SCL}. Finally, we conclude with some open problems in Section~\ref{conclusion}.
	\section{Notations and Definitions}\label{defSect}
	For simplicity of presentation we shall use zero-based numbering. A vector of length $N$ is treated as a row and denoted by bold lowercase letters, such as $\x$ or $\x_0^{N-1}$,
	and the $i$th entry of the vector $\x$ is referred to as $x_i$. 
	Given a binary vector $\x$, we define its support $\s(\x)$ as the set of coordinates in which the vector $\x$ has nonzero entries. Let $\d(\x,\y)$ be the Hamming distance between  $\x$ and $\y$, and $\w(\x)$ be the Hamming weight of $\x$. The set of integers from $i$ to $j-1$, $0\le i< j$, is abbreviated by $[i,j)$ or simply $[j-1]$ if $i=0$. Clearly, $\w(\x)=\d(\x,\0)$, where $\0$ is the all-zero vector. Let $(\textbf{x}, \textbf{y})$ denote the concatenation of two vectors $\textbf{x}$ and $\textbf{y}$. Given an $(N\times N)$ binary matrix $X$ and $\A\subset[0,N)$, we write $X(\A)$ to denote the  $(|\A|\times N)$ submatrix  of $X$ formed by the rows of $X$ with indices in $\A$.    
	
	Let $W : \X \to \Y$ be a B-MSC  with input alphabet
	$\X = \{0, 1\}$, output alphabet $\Y$, and transition probabilities
	$W(y|x)$ for $x \in \X$ and $y \in \Y$. By $W^N$ we denote the vector
	channel corresponding to $N$ independent copies of $W$, i.e., $W^N : \X^N \to
	Y^N$ with transition probabilities 
	$$
	W^N(\y_0^{N-1}| \x_0^{N-1}) = \prod_{i=0}^{N-1}W(y_i | x_i).
	$$	
	Ar{\i}kan used a construction based on the following kernel matrix 
	$$
	G_{2}\defeq\begin{pmatrix} 
	1 &0\\ 
	1 &1
	\end{pmatrix}.
	$$ 
	Given $N= 2^n$, we consider the $(N\times N)$ binary matrix $G_N: = G_2^{\otimes n}$ by performing the $n$th Kronecker power of $G_2$. We denote the $i$th row of $G_N$ by $\g_i$. Usually a linear mapping $\x=\x(\u):\, \X^n\to\X^n$ is defined by  
	$$
	\x = \u B_N G_N,
	$$
	  where $B_N$ is the $(N\times N)$ bit-reversal permutation matrix defined in~\cite[Section VII-B]{arikan2009channel}, and the vectors $\x$, $\u$, and the vector space $\X^n$ are over $GF(2)$. However, since $B_N G_N = G_N B_N$~\cite[Proposition 16]{arikan2009channel}, the latter being a simple permutation on $\x$, we can dispense with $B_N$ in this paper and assume
	\begin{equation}\label{mapping}
	\x = \u G_N.
	\end{equation}
	Let us produce a vector
	channel $W_N : \X^N \to \Y^N$ as follows
	$$
	W_N(\y|\u) \defeq W^N(\y|\u G_N) = W^N(\y|\x).
	$$	
	Given $i\in[0,N)$, we define the synthetic channel $W_N^{(i)}:\X\to\Y^N\times \X^{i}$ as
	$$
	W_N^{(i)}(\y,\u_0^{i-1}|u_i):= \sum\limits_{\u_{i+1}^{N-1}\in\X^{N-i-1}}\frac{1}{2^{N-1}}W_N(\y|\u).
	$$
	\subsection{Polar Coding}
	The generator matrix of a polar code is given by $G_N(\A)$ for some set $\A\subset[0,N)$, which is referred to as the information set. The indices $\A^c :=[0,N)\setminus \A$ are usually called frozen and chosen carefully according to the reliabilities of the synthetic channels~\cite{arikan2009channel}. Namely, in the symmetric channel case, any message $\u\in\{0,1\}^N$ has $u_i=0$ for all $i\in\A^c$, and is mapped to the codeword $\x$ by~\eqref{mapping}.

	Let $\x$
	be sent over $W^N$, and
	let a channel output $\y$
	be received. Given $\A$ and $\y$, the decoder generates an estimate $\hat \u$
	of $\u$. We shall briefly describe SC decoding as the sequential use of the closest coset decoding~\cite{hemmati1989closest}.  
	
	For any binary vector $\v\in\{0,1\}^i$, let the set $C^{(n)}(\v)$ induced by $\v$ be defined as follows
	$$
	C^{(n)}(\v)\defeq\sum_{j\in \s(\v)}\g_{j} +\langle \g_{i},\ldots,\g_{N-1}\rangle,
	$$
	where $\langle\cdot\rangle$ is a linear span of a set of vectors.  By
	\begin{subequations}
		\begin{align*}    
		&C^{(n)}(\v,0)\defeq\sum_{j\in \s(\v)}\g_{j} +\langle \g_{i+1}, \ldots, \g_{N-1}\rangle,\\
		&C^{(n)}(\v,1)\defeq\g_{i} + C^{(n)}(\v,0),
		\end{align*}
	\end{subequations}
	define the zero and the one cosets induced by $\v$, respectively. Obviously, the disjoint union of the zero and the one cosets coincides with $C^{(n)}(\v)$. 
	
	At the beginning of the $i$th stage of SC decoding, we are given  a binary vector $\hat \u^{i-1}_0\in\{0,1\}^{i}$, which can be treated as an estimate of $\u_0^{i-1}$. If $i\in\A^c$, then the decoder makes a bit decision $\hat u_{i} = 0$. Otherwise, the decoder computes the values $W_N^{(i)}(\y,\hat\u_0^{i-1}|0)$ and $W_N^{(i)}(\y,\hat\u_0^{i-1}|1)$, where the value $W_N^{(i)}(\y,\u_0^{i-1}|u_i)$ and the set $C^{(n)}(\u_0^{i-1}, u_i )$ are connected by 
	\begin{equation}\label{calculate probabilities}
	W_N^{(i)}(\y,\u_0^{i-1}|u_i)=\frac{1}{2^{N-1}}\sum\limits_{\v\in C^{(n)}( \u_0^{i-1}, u_i)} W^N(\y|\v).
	\end{equation}
	Then the decoder makes a bit estimate $\hat u_{i}$ of $u_i$: $\hat u_{i} = 0$ if $W_N^{(i)}(\y,\hat\u_0^{i-1}|0)> W_N^{(i)}(\y,\hat\u_0^{i-1}|1)$, and $\hat u_{i} = 1$ if $W_N^{(i)}(\y,\hat\u_0^{i-1}|0)<W_N^{(i)}(\y,\hat\u_0^{i-1}|1)$. For the case $W_N^{(i)}(\y,\hat\u_0^{i-1}|0)= W_N^{(i)}(\y,\hat\u_0^{i-1}|1)$, the decoder chooses the value of  $\hat u_{i}$  randomly and uniformly. This decision rule can be seen as choosing the ``closest'' (zero or one) coset to the received $\y$.  If the wrong coset 
	is selected at some decoding stage, then this decoding error is
	propagated to the next stages.
	\subsection{Weight Distribution} 
	Without loss of generality, we assume that the all-zero codeword is transmitted (e.g., see~\cite[Section VI]{arikan2009channel}, i.e., $\u = \x = \0$. At the $i$th stage, an error occurs if the decoder selects $C^{(n)}(\0_0^{i-1}, 1)$ instead of $C^{(n)}(\0_0^{i-1}, 0)$
	(in the very beginning $C^{(n)}(1)$ instead of $C^{(n)}(0)$, respectively). Let us introduce the weight distribution for SC decoding. For $i\in[0, N)$ and $w\in[N]$, let $S_{i,w}^{(n)}$ be the number of words of weight $w$ in $C^{(n)}(\0_0^{i-1},1)$ and $\S_{i}^{(n)}:=\left(S_{i,w}^{(n)}\right)_{w=0}^{N}$.
	\begin{exmp}
In Table~\ref{tab::weight distributions example} we illustrate weight distributions $\S_{i}^{(n)}$ for $n=3$ and different $i$'s.	Rows in this table correspond to values $i$'s and columns represent weights $w$'s, whereas the entry at position $(i,w)$ is $S_{i,w}^{(n)}$. One can easily check that the weight distributions are symmetric in $w$, i.e. $S_{i,w}^{(n)}=S_{i,N-w}^{(n)}$ for all $i$ except the case $i=N - 1$. This holds since $\g_{N-1}$ is the all-one vector.
	\begin{table}
		\caption{Weight distributions $\S_{i}^{(n)}$ for $n=3$}
		\begin{center}
			\begin{tabular}{|c||c|c|c|c|c|c|c|c|c|}
				\hline
				\multirow{ 2}{*}{Rows $i$}	& \multicolumn{9}{c|}{Weights $w$}\\
				\cline{2-10}
				& 0 & 1 & 2 &  3& 4 & 5 & 6 & 7 & 8 \\ 
				\hhline{|=||=|=|=|=|=|=|=|=|=|} 
				0& 0 & 8 & 0 & 56 & 0 & 56 & 0 & 8 & 0 \\ 
				\hline 
				1& 0 & 0 & 16 & 0 & 32 & 0 &  16 &  0&  0\\ 
				\hline 
				2& 0 & 0 & 8 & 0 & 16 & 0 &  8 &  0&  0\\ 
				\hline 
				3& 0 & 0 & 0 & 0 & 16 & 0 & 0 & 0 & 0 \\ 
				\hline 
				4& 0 & 0 & 4 & 0 & 0 & 0 & 4 & 0 & 0 \\ 
				\hline 
				5& 0 & 0 & 0 & 0 & 4 & 0 & 0 & 0 & 0 \\ 
				\hline 
				6& 0 & 0 & 0 & 0 & 2 & 0 & 0 & 0 & 0 \\ 
				\hline 
				7& 0 & 0 & 0 & 0 & 0 & 0 & 0 & 0 & 1 \\ 
				\hline 
			\end{tabular} 
		\label{tab::weight distributions example}
		\end{center} 
	\end{table}
\end{exmp}
\subsection{Approximate Upper Bound on the Error Probability}
The weight distribution is useful for obtaining
upper bounds of the error probability. One approximate bound is introduced in this subsection. Let $P_e(i)$ be the error probability at the $i$th decoding stage, i.e., 
\begin{align*}
P_e(i)&=\Pr\left\{W_N^{(i)}(\y,\0_0^{i-1}|0) < W_N^{(i)}(\y,\0_0^{i-1}|1)\right\}\\&+\Pr\left\{W_N^{(i)}(\y,\0_0^{i-1}|0) = W_N^{(i)}(\y,\0_0^{i-1}|1)\right\}/2.
\end{align*}
 This implies by~\eqref{calculate probabilities}
$$
P_e(i)\le \Pr\left\{W^N(\y|\0)\le\sum\limits_{\v\in C^{(n)}(\0_0^{i-1}, 1)} W^N(\y|\v)\right\}.
$$
	 For the binary phase-shift keying (BPSK) transmission over the additive white Gaussian noise (AWGN) channel with small enough variance $\sigma^2$, the right-hand side of the inequality above can be well-approximated by 
	 $$
	 \Pr\left\{W^N(\y|\0)\le \max\limits_{\v\in C^{(n)}(\0_0^{i-1}, 1)} W^N(\y|\v)\right\},
	 $$
	 whereas the latter is upper bounded~\cite[Equations (4.109) and (4.110)]{jacobs1965principles} by the union bound
	\begin{equation}\label{union bound}
	P_{ub}(i)\defeq\sum\limits_{w = 1}^N \frac{1}{2} S_{i,w}^{(n)} \err\left(\sqrt{w/(2\sigma^2)}\right).
	\end{equation}
	Here, $\err(\cdot)$ is the complementary error function  defined by
	$$
	\err(x)\defeq\frac{2}{\sqrt{\pi}}\int_{x}^{\infty}e^{-t^2}dt.
	$$
	 
	It is worth noting that there are several techniques allowing to calculate $P_e(i)$ with inherent inaccuracy and to bound $P_e(i)$. Among them are density evolution (DE)~\cite{mori2009performance}, degrading and upgrading algorithms~\cite{tal2013construct} and Gaussian approximation~\cite{trifonov2012efficient}. 
	\section{Weight Distribution for Successive Cancellation Decoding}\label{wd for sc}
	In this section, we first provide an algorithm for efficiently computing the WDs for  SC decoding. After that, we derive a short formula for computing the first nonzero component of the WD. Finally, a link between the partial order and the first nonzero component of the WD is shown.
	\subsection{Algorithm for Computing the Weight Distributions}
	Our analysis in this subsection is similar to one in~\cite[Section $2$]{fossorier1997weight}, where the WD for the closest coset decoding of $|u|u+v|$ construction was established.
	Let us determine $S_{i,w}^{(n)}$, the number of words of weight $w$ in $C^{(n)}(\0_0^{i-1},1)$. First, we recall that $G_N = G_2\otimes G_{N/2}$ with $N/2 = 2^{n-1}$, i.e., 
	$$
		G_N = \begin{pmatrix}
		G_{N/2} & 0 \\
		G_{N/2} & G_{N/2}
	\end{pmatrix} .
	$$
	
	 So if $i\ge 2^{n-1}$, then any word in $C^{(n)}(\0_0^{i-1},1)$ represents a repetition of some word in $C^{(n-1)}(\0_0^{i-1-N/2},1)$. Thus, $S_{i,w}^{(n)}=0$ for odd $w$, and $S_{i,w}^{(n)}=S_{i-N/2,w/2}^{(n-1)}$ for even $w$.  If $i< 2^{n-1}$, then any word $\x\in C^{(n)}(\0_1^{i-1},1)$ can be uniquely represented in the form
	\begin{equation}\label{basic lemma 1}
	\x = \left(\g_i+\sum\limits_{j\in I_{1}}\g_j\right)+\sum\limits_{j\in I_{2}}\g_j = (\x_1,\0) + (\x_2,\x_2),
	\end{equation}
	where the index sets $I_{1}\subset [i+1,2^{n-1})$ and $I_{2}\subset [2^{n-1},2^{n})$, and $\x_1\in C^{(n-1)}(\0_0^{i-1},1)$ and $\x_2\in\{0,1\}^{N/2}$. Moreover, since $G_{N/2}$ is full-rank, any $\x_2 \in \{0, 1\}^{N/2}$ can be represented as a sum of $\g_j$'s. So, any combination of $\x_1\in C^{(n-1)}(\0_0^{i-1},1)$ and $\x_2\in\{0,1\}^{N/2}$ gives a unique $\x\in C^{(n)}(\0_0^{i-1},1)$. Hence, $S_{i,w}^{(n)}$ can be determined using the following statement.
	\begin{theorem}\label{th::key theorem weight distribution}
		For $t\in\{0,\ldots,N/2-w'\}$, the contribution of $\x_1$ with $\w(\x_1)=w'$ to $S^{(n)}_{i,w'+2t}$ is $2^{w'}\binom{N/2 - w'}{t}$. 
	\end{theorem}
	\begin{proof}[Proof of Theorem~\ref{th::key theorem weight distribution}]
		Let $\x \in C^{(n)}(\textbf{0}^{i-1}_0, 1)$. By~\eqref{basic lemma 1}, it is easy to check that 
		\begin{align}
		\w(\x)= &\,\w(\x_2) + \d(\x_1,\x_2)=
		\w(\x_1) \notag \\ 
		+&\, (\w(\x_2) + \d(\x_1,\x_2)-\w(\x_1))\ge \w(\x_1). \label{basic lemma 2}
		\end{align}
		We observe that the sum in the parentheses is equal to double the number of coordinates $i$ so that $x_{1,i}=0$ and $x_{2,i}=1$. Given $\x_1$ with $\w(\x_1)=w'$, there are $\binom{N/2-w'}{t}$ different choices for placing $t$ ones in $\x_2$ among $N/2 - w'$ coordinates corresponding to zeros in $\x_1$.  By~\eqref{basic lemma 1}, $\x_2$ could have anything in the remaining $w'$ coordinates corresponding to ones in $\x_1$. Therefore the total number of choices for $\x_2$ is $2^{w'}\binom{N/2 - w'}{t}$. 
	\end{proof} 
	Summarizing the arguments given above, the WDs can be calculated in a recursive manner with the help of Algorithm~\ref{algorithm for computing}. We start the algorithm by initializing  $\S_0^{(0)}=(0,1)$ as for the base case, we have $G_0=(1)$ and $C^{(0)}(1)=\{(1)\}$.
	\begin{rmrk} \label{rem::computing wd for zero coset}
		 Let $T^{(n)}_{i,w}$ denote the number of words of weight $w$ in $C^{(n)}(\0_0^{i-1},0)$. 
		 The WDs for the zero coset can be obtained by Algorithm~\ref{algorithm for computing} by only changing the initialization step from $\S_0^{(0)}:=(0,1)$ to $\T_0^{(0)}:=(1,0)$ as $C^{(0)}(0)=\{(0)\}$. Let us briefly check this. If $i\ge 2^{n-1}$, then any word in $C^{(n)}(\0_0^{i-1},0)$ represents a repetition of some word in $C^{(n-1)}(\0_0^{i-1-N/2},0)$. Thus, $T_{i,w}^{(n)}=0$ for odd $w$, and $T_{i,w}^{(n)}=T_{i-N/2,w/2}^{(n-1)}$ for even $w$.  If $i< 2^{n-1}$, then any word $\x\in C^{(n)}(\0_1^{i-1},0)$ can be uniquely represented in the form
		\begin{equation*}
		\x = \left(\sum\limits_{j\in I_{1}}\g_j\right)+\sum\limits_{j\in I_{2}}\g_j = (\x_1,\0) + (\x_2,\x_2),
		\end{equation*}
		where the index sets $I_{1}\subset [i+1,2^{n-1})$ and $I_{2}\subset [2^{n-1},2^{n})$, and $\x_1\in C^{(n-1)}(\0_0^{i-1},0)$ and $\x_2\in\{0,1\}^{N/2}$. Moreover, 	for $t\in\{0,\ldots,N/2-w'\}$, the contribution of $\x_1$ with $\w(\x_1)=w'$ to $T^{(n)}_{i,w'+2t}$ is again $2^{w'}\binom{N/2 - w'}{t}$.	
	\end{rmrk}
	\begin{algorithm}[t]
		\caption{Computing the weight distributions}
		\label{algorithm for computing}
		\begin{algorithmic}[1]
			\renewcommand{\algorithmicrequire}{\textbf{Input:}}
			\renewcommand{\algorithmicensure}{\textbf{Output:}}
			\REQUIRE length $N=2^n$
			\ENSURE  weight distributions $\left\{\S_{i}^{(n)}:\,i\in[0,N)\right\}$
			\\ \textit{Initialization} :
			\STATE     $\S_0^{(0)}\gets \left(S_{0,0}^{(0)}, S_{0,1}^{(0)} \right)$ with $S_{0,0}^{(0)}\gets 0$ and $S_{0,1}^{(0)}\gets 1$
			\FOR{Kronecker's power $j = 1$ to $n$}
			\STATE
			$N_j\gets 2^j$
			\FOR{row index $i = 0$ to $N_j-1$}\label{al::main loop}
			\IF {($i < N_{j}/2$)}
			\FOR{weight $w = 0$ to $N_j$} \label{al::first for loop}
			\STATE
			$S^{(j)}_{i,w}\gets \sum\limits_{\substack{0\le w'\le \min(w,N_{j}/2) \\ w' \equiv  w\pmod{2}}} S^{(j-1)}_{i,w'} \binom{N_{j}/2-w'}{(w-w')/2} 2^{w'}$ \label{al::formula first part}\\
			\ENDFOR
			\ELSE
			\FOR{weight $w = 0$ to $N_j/2$} \label{al::second for loop}
			\STATE $S^{(j)}_{i,2w}\gets S^{(j-1)}_{i-N_j/2,w}$ \label{al::formula second part}\\
			$S^{(j)}_{i,2w+1}\gets 0$
			\ENDFOR
			\ENDIF
			\ENDFOR
			\ENDFOR
			\RETURN $\left\{\S_{i}^{(n)}:\,i\in[0,N)\right\}$
		\end{algorithmic} 
	\end{algorithm}

	\begin{rmrk}  Algorithm~\ref{algorithm for computing} provides a practical way to determine WDs associated with SC decoding.  Let us take some $p\ge 1$ and assume that we have a precomputed look-up table of appropriate binomial coefficients. Then the complexity of computing the first $p$ nonzero components of $\S_{i}^{(n)}$ for all $i\in[0,N)$ is $O(p^2 N)$. Indeed, the complexity of the for-loops on lines~\ref{al::first for loop} and~\ref{al::second for loop} can be reduced to $O(p^2)$  and $O(p)$, respectively,  if only the first $p$ nonzero components are computed. Now we check this claim for line~\ref{al::first for loop}. Suppose that the first nonzero component of $\S^{(j-1)}_i$ corresponds to weight $w_i^{(j-1)}$. By line~\ref{al::formula first part}, the weights of the first $p$ nonzero components of $\S_i^{(j)}$ are between $w_i^{(j-1)}$
		and $w_i^{(j-1)}+ 2p - 2$ and to compute one of them
		we need to sum at most $p$ numbers. Hence, the loop on line 6 has complexity $O(p^2)$.
	 Similar arguments work out for line~\ref{al::formula second part}. The total number of required computations of the algorithm is also linear with the number of times when we are in the for-loop on line~\ref{al::main loop}, which is equal to $2N-2$. This implies that for $p=O(1)$, the
	complexity of getting approximate upper bounds on the error probabilities for all the subchannels based on the formula~\eqref{union bound}, restricted to the first  $p$ nonzero terms, is $O(N)$.
	\end{rmrk}
	\begin{exmp}
		Let us illustrate the bound~\eqref{union bound} by taking code length $N=256$ and the AWGN channel with variance $\sigma^2=0.158$ (SNR is around $8.01$ dB). Using Algorithm~\ref{algorithm for computing} we compute the weight distributions and depict the pairs $(P_{e}(i),P_{ub}(i))$  in Figure~\ref{Bounding Bad Probability of Channel 1}, where the approximate union bound $P_{ub}(i)$  on the decoding error probability $P_{e}(i)$ is computed with the help of~\eqref{union bound}. We compute $P_{e}(i)$ with the help of DE. 
		\begin{figure}
			\centering
			\begin{minipage}{.49\textwidth}
			\begin{tikzpicture}
			\begin{loglogaxis}[xlabel={$P_e(i)$}, ylabel={$P_{ub}(i)$}, xmin=1e-10, xmax=1, ymin=1e-10, ymax=1]
			
			\addplot[smooth, thick, dotted, color=black] coordinates {
				(1e-10, 1e-10)
				(1,1)
			};
			
			\foreach \Point in {(0.6,0.6),(0.4,0.6),(0.3,0.6),(0.21,0.51),(0.08,0.1),(0.12,0.17),(1.5e-10,1.5e-10),(3e-10,3e-10),(8e-10,8e-10),(1.2e-9,1.2e-9),(2e-10,2e-9),(2e-9,2e-9),(1.4e-9,2e-9),(3.4e-9,4e-9),(4.5e-9,4.5e-9),(1.2e-9,1.5e-8),(1e-8,1e-8),(0.8e-8,1e-8),(2.1e-8,2.1e-8),(1.6e-8,2.05e-8),(4.3e-8,4.9e-8),(5e-8,5e-8),(8e-8,9e-8),(9e-8,10e-8),(0.8e-7,8e-7),(2e-7,2e-7),(1.4e-7,1.9e-7),(2e-7,4e-7),(3e-7,4e-7),(5e-7,6.5e-7),(7e-7,7.5e-7),(1e-6,1.5e-6),(1.1e-6,1.5e-6),(1.5e-6,2.5e-6),(3e-6,4e-6),(4e-6,7e-6),(8e-6,2e-5),(3e-5,3e-5),(4e-5,4e-5),(9e-5,9e-5),(1e-4,1e-4),(2e-4,2e-4),(2.1e-4,2.1e-4),(2.15e-4,2.15e-4),(4e-4,4e-4),(4.2e-4,4.2e-4),(1e-3,1e-3),(0.8e-3,1e-3),(2e-3,3e-3),(2.5e-3,3e-3),(4e-3,5e-3),(2e-2,2e-2),(4e-2,4e-2),(2e-1,5e-1),(1e-1,4e-1),(0.7e-1,3e-1),(0.4e-1,2e-1),(0.5e-1,2.5e-1),(0.2e-1,1e-1),(0.25,4.5e-1),(2e-2,1.1e-1),(1e-2,8e-2),(0.8e-2,5e-2),(0.7e-2,2e-2),(0.5e-2,1.8e-2),(0.5e-2,1e-2),(0.4e-2,0.9e-2),(1e-5,5e-5),(1.5e-5,7e-5),(3e-5,9e-5),(4e-5,2e-4),(1e-4,7e-4),(2e-4,8e-4),(4e-4,1e-3),(5e-4,1.7e-3),(0.8e-5,3e-5),(0.6e-5,1.2e-5),(1e-3,1.4e-3),(1e-6,1.7e-6),(5e-7,9e-7),(7e-6,9e-6),(1e-4,4e-4),(0.8e-4,3.7e-4)}{
				\addplot[mark=o, color=NavyBlue] coordinates{\Point};
			}
			\end{loglogaxis}
			\end{tikzpicture}
			\caption[My plot]{The approximate union bound  $P_{ub}(i)$ and the bit-channel probabilities of error $P_e(i)$. 
				Only those pairs $(P_e(i),P_{ub}(i))$ for which the error probability $P_e(i)$ is greater than $10^{-10}$
				threshold are shown. Additionally, the dashed line $(x,x)$ for $x\in[10^{-10},10^0]$ is plotted.}
			\label{Bounding Bad Probability of Channel 1}
		\end{minipage}
	\begin{minipage}{.49\textwidth}
		\begin{tikzpicture}
		\begin{semilogyaxis}[xlabel={Signal-to-noise ratio [dB]}, ylabel=Bit Error Rate, xmin=3, xmax=8, ymin=0.00001, ymax=1, grid=both,legend style={font=\small}]
		
		\addplot[thick,smooth, color=red] coordinates {
			(3, 0.2)
			(4, 0.07)
			(5, 0.02)
			(6, 0.004)
			(7, 0.0004)
			(8, 0.00003)
		};
		\addlegendentry{DE}
		\addplot[smooth, thick, mark=o, color=NavyBlue] coordinates {
			(3, 0.8)
			(4, 0.15)
			(5, 0.03)
			(6, 0.0045)
			(7, 0.00041)
			(8, 0.00003)
		};
		\addlegendentry{UB}
		\end{semilogyaxis}
		\end{tikzpicture}
		\caption[Union and sphere bounds on $P_e(72)$ under SC]{The DE computations and the approximate union bound on $P_e(72)$ under SC decoding.}
		\label{Bounding Bad Probability of Channel 2}
		\end{minipage}
		\end{figure}

	\end{exmp}

	\begin{exmp}
	Now let us take the code length $N=128$ and the synthetic channel with index $i=72$. Using Algorithm~\ref{algorithm for computing} we compute the weight distributions and depict $P_{ub}(72)$ the approximate union bound~\eqref{union bound} on the decoding error probability $P_{e}(72)$ along with this probability, calculated with the help of DE, in Figure~\ref{Bounding Bad Probability of Channel 2}. 
\end{exmp}	
	\subsection{First Nonzero Component}
	Let  $s_i^{(n)}$ be the first nonzero component of $\S_{i}^{(n)}$. It is known (e.g., see~\cite[Chapter 6]{korada2009phd}) that the first nonzero component of $\S_{i}^{(n)}$ corresponds to weight $\w(\g_i)$ and, thus, $s_i^{(n)} = S_{i, \w(\g_i)}^{(n)}>0$, i.e.,
		$$
	\S_{i}^{(n)}=(\underbrace{0,0,\ldots, 0}_{\w(\g_i)},s_i^{(n)},\ldots).
		$$
	We extend this line of research and find an explicit formula for $s_{i}^{(n)}$. Given $j\in[0,n)$ and $i\in[0,N)$, let $b_j(i)$ be the $j$th bit in the binary representation of integer $i$, and $p_j(i)$ be the partial sum of the first $j+1$ bits, i.e., integers $i$ and $p_j(i)$ can be represented as 
	\begin{align*}
	i =\sum\limits_{j=0}^{n-1} b_j(i) 2^j, \quad p_j(i)&=\sum\limits_{s=0}^{j} b_s(i).
	\end{align*}
	\begin{theorem}\label{th::first component}
		Given $N = 2^n$ and $i\in[0,N)$, the first nonzero component of $\S_{i}^{(n)}$ corresponds to weight $\w(\g_i)=2^{p_{n-1}(i)}$ and equals $s_{i}^{(n)}$, where 
		\begin{equation}\label{formula for first nonzero component}
		\log_2 s_{i}^{(n)} = \sum_{j=0}^{n-1} (1 - b_j(i))2^{p_j(i)}.
		\end{equation}
	\end{theorem}
	\begin{proof}[Proof of Theorem~\ref{th::first component}]
		We shall prove the statement of this theorem by induction on $n$. The base case is evident as $S^{(0)}_{0,1}=1$. We assume for a moment that the equality~\eqref{formula for first nonzero component} holds for some pair $(n,i)$ with  $i\in[0,N)$ . Let $\g_i$ be the $i$th row in the matrix $G_{N}$. We derive the statement of this theorem for $(n+1,i)$ and $(n+1,i+N)$.
		
		 First, let us prove the case $(n+1,i)$. By  line~\ref{al::formula first part} of Algorithm~\ref{algorithm for computing}, 	for $i < N=2^n$, we have
		 \begin{equation}\label{contribution to nonzero component}
		 S^{(n+1)}_{i,w}: = \sum\limits_{\substack{0\le w'\le \min(w,N) \\ w' \equiv  w\pmod{2}}} S^{(n)}_{i,w'} \binom{N-w'}{(w-w')/2} 2^{w'}.
		 \end{equation}
		Therefore, we get that $S_{i,w}^{(n+1)} =0$ for all $w< \w(\g_i)$ since $S_{i,w'}^{(n)} =0$ for all $w'< \w(\g_i)$. For $w=\w(\g_i)$, the nonzero contribution to $S^{(n+1)}_{i,w}$ in the right-hand side of~\eqref{contribution to nonzero component}  comes from the only term indexed by $w'=w$. Thus, we have  $s_i^{(n+1)}=s_i^{(n)}2^{\w(\g_i)}$ and since $b_n(i)=0$ so that $2^{p_n(i)}=2^{p_{n-1}(i)}=\w(\g_i)$,
		\begin{align*}
		\log_2 s_i^{(n+1)} &= \log_2 s_i^{(n)} + \w(\g_i)\\
		 &= \sum_{j=0}^{n-1} (1 - b_j(i))2^{p_j(i)} + \w(\g_i)  \\
		 &= \sum_{j=0}^{n} (1 - b_j(i))2^{p_j(i)}.
		\end{align*}
		  Similarly, we deal with the case $(n+1,i+N)$. Let $i':=i+N\ge N$.  According to line~\ref{al::formula second part} of Algorithm~\ref{algorithm for computing}, we have
	  \begin{equation}\label{eq::second part evaluation}
	  	  S_{i',2w}^{(n+1)}=S^{(n)}_{i,w},\quad S^{(n+1)}_{i',2w+1} = 0.
	  \end{equation}
	  We obtain $S_{i',w}^{(n+1)}=0$ for all $w<2 \w(\g_i)$ since $S_{i,w}^{(n)} =0$ for all $w< \w(\g_i)$. Thus, the first nonzero component of $\S_{i'}^{(n+1)}$ corresponds to the weight $2\w(\g_i)=\w(\g_{i'})$, where $\g_{i'}$ is the $i'$th row in the matrix $$
	  G_{2N}=\begin{pmatrix}G_{N} & 0 \\ G_{N} & G_{N} \end{pmatrix}.
	  $$
	    Moreover, by~\eqref{eq::second part evaluation} we have $s_{i'}^{(n+1)}=s_i^{(n)}$ and since $b_n(i') = 1$
	  \begin{align*}
	  \log_2s_{i'}^{(n+1)}&=\log_2s_{i}^{(n)}\\
	  &=\sum_{j=0}^{n-1} (1 - b_j(i))2^{p_j(i)}\\&=\sum_{j=0}^{n} (1 - b_j(i'))2^{p_j(i')}.
	  \end{align*}
	  This proves the inductive step and completes the proof.
	\end{proof}
	\subsection{Connection with the Partial Order}
	It was observed~\cite{schurch2016partial,bardet2016algebraic} that there is
	a partial order between the synthetic channels, which holds for
	any B-MSC. Let us rephrase  this result using our notation.
	\begin{theorem}[The partial order {\cite[Definition $8$]{schurch2016partial}}] \label{th::partial order}
		 $W^{(i)}_N$ is stochastically degraded by $W^{(j)}_N$ if there exists a finite sequence of integers $a_0,a_1,\ldots,a_\ell\in [0,N)$, $\ell \ge 0$, such that $a_0 = i$, $a_\ell = j$ and for all $m \in [0,\ell)$, one of the following two properties holds:
		 \begin{enumerate}
		 	\item There exist two indices $u=u(m),w=w(m)\in[0,n)$ so that $u<w$ and
		 	\begin{align*}
		 	b_k(a_m)&= b_k(a_{m+1}) \text{ for all }k\in[0,n)\setminus\{u,w\},\\
		 	b_u(a_m) &= b_w(a_{m+1})= 1,\\
		 	b_w(a_m) &= b_u(a_{m+1}) = 0;
		 	\end{align*}
		 	\item  $b_k(a_m)\le b_k(a_{m+1})$ for all $k\in[0,n)$.
		  \end{enumerate}		 
	\end{theorem}
Note that for any such sequence $a_0,a_1,\ldots,a_\ell$, we have the property
$$
\sum\limits_{k=n-1-t}^{n-1} b_k(a_m) \le \sum\limits_{k=n-1-t}^{n-1} b_k(a_{m+1}) \
$$  for all $t\in[0,n),\,m\in[0,\ell)$. In particular, this means  
\begin{equation}\label{technical property for PO}
\sum\limits_{k=n-1-t}^{n-1} b_k(i) \le \sum\limits_{k=n-1-t}^{n-1} b_k(j) \text{ for all }t\in[0,n).
\end{equation}
Theorem~\ref{th::first component} shows us a natural one-way connection between the partial order  given in Theorem~\ref{th::partial order} and the first nonzero components of WDs. Let $i<j$ be fixed. If the $i$th synthetic channel is stochastically degraded by the $j$th one by the partial order, then the property~\eqref{technical property for PO} holds. In particular, this implies that $p_{n-1}(i)\le p_{n-1}(j)$. In case the latter inequality is strict, by Theorem~\ref{th::first component}, we deduce that $\w(\g_i)<\w(\g_j)$ and the first nonzero component of $\S_{j}^{(n)}$ corresponds to a larger weight than  the first nonzero component of $\S_{i}^{(n)}$:
\begin{align*}
	\S_{i}^{(n)}=&(\underbrace{0,\ldots, 0}_{\w(\g_i)},s_i^{(n)},\ldots),\\
	\S_{j}^{(n)}=&(\underbrace{0,\ldots, 0}_{\w(\g_i)},\underbrace{0,\ldots, 0}_{\w(\g_j)-\w(\g_i)},s_j^{(n)},\ldots).	
\end{align*}	

In case $p_{n-1}(i)= p_{n-1}(j)$, we have $\w(\g_i)=\w(\g_j)$ and the inequality~\eqref{technical property for PO} can be rewritten in a simple way 
\begin{align*}\label{f:ps}
p_t(i)&=p_{n-1}(i)-\sum\limits_{k=t+1}^{n-1} b_k(i) \\&\ge p_{n-1}(j)-\sum\limits_{k=t+1}^{n-1} b_k(j)=p_t(j) 
\end{align*}
for all $t\in[0,n)$.
Now we shall prove that $s_i^{(n)}>s_j^{(n)}$. To this end,  it is enough to show that for $a_m$ and $a_{m+1}$ that fulfill the first property in Theorem~\ref{th::partial order}, we have $s_{a_m}^{(n)} > s_{a_{m+1}}^{(n)}$. Indeed, if two distinct integers $a_m$ and $a_{m+1}$ satisfy the second property in Theorem~\ref{th::partial order}, then 
$$
p_{n-1}(i)\le p_{n-1}(a_m)<p_{n-1}(a_{m+1})=p_{n-1}(j).
$$
By Theorem~\ref{th::first component}, we know
\begin{align*}
\log_2 s_{a_m}^{(n)} =& \sum_{k=0}^{n-1} (1 - b_k(a_m))2^{p_k(a_m)},\\
\log_2 s_{a_{m+1}}^{(n)} =& \sum_{k=0}^{n-1} (1 - b_k(a_{m+1}))2^{p_k(a_{m+1})}.
\end{align*}
By the first property in Theorem~\ref{th::partial order}, 	$b_k(a_m)= b_k(a_{m+1})$ for all $k\in[0,n)\setminus\{u,w\}$ with $u<w$ and $b_u(a_m)=b_w(a_{m+1})=1$ and $b_w(a_m)=b_u(a_{m+1})=0$.  Thus, we obtain
\begin{align*}
\log_2 s_{a_m}^{(n)} &- \log_2 s_{a_{m+1}}^{(n)} \\&= \sum_{k=u+1}^{w} (1 - b_k(a_m))2^{p_k(a_m)} \\&- \sum_{k=u}^{w-1} (1 - b_k(a_{m+1}))2^{p_k(a_{m+1})}\\&=2^{p_w(a_m)} - 2^{p_u(a_{m+1})}\\&+ \sum_{k=u+1}^{w-1} (1 - b_k(a_m))(2^{p_k(a_m)}-2^{p_k(a_{m+1})}) >0,
\end{align*}
where we used the property that $p_k(a_m) = p_k(a_{m+1})+1$ for $k\in[u+1,w)$ and $p_w(a_m)\ge p_u(a_m) = p_u(a_{m+1})+1$.
Combining the arguments above we arrive to the following statement.
\begin{theorem}
 Let $W^{(i)}_N$ be stochastically degraded by  $W^{(j)}_N$ with $i<j$ by Theorem~\ref{th::partial order}. Then one of the two statements holds:
 \begin{enumerate}
 	\item The first nonzero component of $\S_j^{(n)}$  corresponds to a larger weight than the first nonzero component of $\S_i^{(n)}$, i.e., $\w(\g_j)>\w(\g_i)$.
 	\item The first nonzero component of $\S_j^{(n)}$  corresponds to the same weight as the first nonzero component of $\S_i^{(n)}$, i.e., $w:=\w(\g_j)=\w(\g_i)$, and we have $s_i^{(n)}>s_j^{(n)}$. 
 \end{enumerate}
\end{theorem}
	\section{Toward Weight Distribution for Successive Cancellation List Decoding}\label{wd for SCL}
	Let us briefly recall the high level description of  the successive cancellation  list (SCL) decoder~\cite{tal2015list} with list size $L$. At the $i$th decoding stage for $i\in\A$, we split each path $\hat\u_0^{i-1}$ from the list of candidates, abbreviated by $\L$, 
	into two paths by taking $\hat u_{i}=0$ and $\hat u_{i}=1$ and calculate two values $W_N^{(i)}(\y,\hat\u_0^{i-1}|0)$ and $W_N^{(i)}(\y,\hat\u_0^{i-1}|1)$ by~\eqref{calculate probabilities}.
	Since the number of paths is doubled, we keep in $\L$ only  the $L$ most likely paths at each stage. The pruning criterion is based on the values 
	$$
	\left\{W_N^{(i)}(\y,\hat\u_0^{i-1}|0),\ W_N^{(i)}(\y,\hat\u_0^{i-1}|1)\right\}_{\hat\u_0^{i-1}\in\L}.
	$$
	 If index $i\in\A^c$, then for any path $\hat \u_0^{i-1}$ from $\L$, the decoder makes a bit decision $\hat u_i = 0$ and keeps this $\hat\u_0^{i}$ in $\L$.
	
	Assume that after the $j$th decoding stage, the SCL decoder keeps (at least) the following two paths: true path $\0_0^j$ and path $\u_0^j$ mistaken in only two positions $i$ and $j$, $i<j$, i.e., $\s(\u_0^j)=\{i,j\}$. Our goal is to estimate the minimal distance between sets induced by these two paths. For simplicity of notation we abbreviate $C^{(n)}(\u_0^{j})$ by $C^{(n)}(i,j)$.  Also, recall that
	$$
	C^{(n)}(i,j) = \g_{i} + \g_{j} +\langle \g_{j+1}, \ldots, \g_{N-1}\rangle.
	$$
	\begin{theorem}\label{theorem on minimal distance}
	Let $N = 2^n$. For any $i, j \in [0, N)$ such that $i < j$ and any $\x \in C^{(n)}(i,j)$, we have $\w(\x) \geq \w(\g_i + \g_j)$, where
		\begin{equation}\label{eq::weight of sum}
		\w(\g_i+\g_j)= \w(\g_i)+\w(\g_j) - 2^{t_{i,j}+1}
		\end{equation}
		with $t_{i,j}:= \sum_{k=0}^{n-1}b_k(i) b_k(j)$.
	\end{theorem}
Before we start proving the theorem, let us introduce several useful definitions. Recall the notion $b_r(k)$, the $r$th bit in the binary representation of integer $k$. Let $\gg_0:=(1,0)$ and $\gg_1:=(1,1)$. In other words, the first row of $G_2$ is $\gg_0$ and the second one is $\gg_1$.  By induction, any row $\g_k$ in the matrix $G_N=G_2^{\otimes n}$ can be shown to be a Kronecker product of vectors
\begin{equation}\label{eq::representation of gi}
\g_k=\gg_{b_{n-1}(k)}\otimes \gg_{b_{n-2}(k)}\otimes \ldots \otimes \gg_{b_{0}(k)}.
\end{equation} 

Let $\ell$ be an arbitrary integer such that $0\le\ell<n$. By $I_\ell\subset [0,N)$ denote the collection of indices $k\in [0,N)$ such that $b_\ell(k)=1$. Define the complement of $I_\ell$ by $I_\ell^c\defeq[0,N)\setminus I_\ell$. We note that $|I_\ell| = |I_\ell^c|= N/2$. 

Define the function $f_\ell: [0,N)\to [0,N/2)$ that maps an integer $k\in[0,N)$ to the integer $f_\ell(k)$ by removing the $\ell$th bit in the binary representation of $k$
$$
f_\ell(k):=\sum_{r=0}^{\ell-1} b_r(k)2^r + \sum_{r=\ell+1}^{n-1} b_r(k)2^{r-1}.
$$
Note that in the second sum, $b_r(k)$ is multiplied by $2^{r-1}$ and not $2^{r}$. Moreover, in the
special case where $\ell = n - 1$, we have $f_{n-1}(k) = k - b_{n-1}(k) 2^{n-1}$.

Let vector $\x$ be of length $N$ and let $I \subset [0, N)$ be a set of indices in ascending order.
Then, $\x|_{I}$ is a projection of $\x$ onto the coordinates indexed by $I$. That is, $\x|_{I}$ is a vector of length $|I|$ in which any index of $\x$ that is not in $I$ is removed. For example, if $N = 8$, $\x = (0, 0, 1, 0, 1, 1, 0, 1)$, and $I = \{1, 3, 5, 7\}$, then $\x|_{I} = (0, 0, 1, 1)$.
	
Denote $Q:=G_{N/2}$ with rows $\left\{\q_k,\,k\in[0,N/2)\right\}$. 
From the structure~\eqref{eq::representation of gi} and the definition of $I_\ell^c$, we note that
\begin{align}
\g_{k}|_{I_\ell^c}&=\gg_{b_{n-1}(k)}\otimes \ldots\otimes\gg_{b_{\ell+1}(k)}\otimes 1 \otimes\ldots \otimes \gg_{b_{0}(k)}\nonumber\\&=\q_{f_\ell(k)}. \label{eq::important 1}
\end{align}
 Similarly, from the structure~\eqref{eq::representation of gi} and the definition of $I_\ell$, we get that
\begin{equation}\label{eq::important 2}
\g_{k}|_{I_\ell}=\begin{cases}
\0_0^{N/2-1}\quad&\text{if } b_\ell(k)=0,\\
\q_{f_\ell(k)}\quad&\text{if } b_\ell(k)=1.
\end{cases}
\end{equation}
Observe that by~\eqref{eq::important 1}-\eqref{eq::important 2}, $\w(\g_k|_{I^c_\ell}) = \w(\g_k)$  if $b_\ell(k) = 0$. 

The following example demonstrates the definitions and
concepts introduced above.
\begin{exmp}	
	Let $N = 8$, $\ell=0$, $i=2= 0\cdot2^2+1\cdot2^1+0\cdot2^0$ and $j = 5=1\cdot2^2+0\cdot2^1+1\cdot2^0$. Then the vectors $\g_i$ and $\g_j$ can be represented using~\eqref{eq::representation of gi} as
	\begin{align*}
	\g_i=\gg_0 \otimes \gg_1 \otimes \gg_0=&(1,0,1,0,0,0,0,0),\\ \g_j= \gg_1 \otimes \gg_0 \otimes \gg_1 =&(1,1,0,0,1,1,0,0).
	\end{align*}
	Similarly, we have that $\q_1 = \gg_0 \otimes \gg_1 = (1, 1, 0, 0)$ and $\q_2 = \gg_1 \otimes \gg_0 = (1, 0, 1, 0)$.
	The set $I_\ell=\{1,3,5,7\}$  consists of all integers in $[0, N)$ whose $\ell$th bit (least significant bit in this case) in the binary representation is one $(1)$.  The set $I^c_\ell=[0,8)\setminus I_\ell=\{0,2,4,6\}$.   One can see that 
\begin{align*}
f_\ell(0)&=f_\ell(1)=0, \quad f_\ell(2)=f_\ell(3)=1,\\
f_\ell(4)&=f_\ell(5)=2,\quad f_\ell(6)=f_\ell(7)=3.
\end{align*}
	By~\eqref{eq::important 1} and~\eqref{eq::important 2} the restrictions of $\g_i$ and $\g_j$ to coordinates indexed by $I_\ell^c$ and $I_\ell$ are the vectors
	\begin{align*}
	\g_i|_{I_\ell^c}&=(1,0)\otimes(1,1)\otimes(1)
	 = \q_{f_\ell(i)}=\q_1,\\ \g_i|_{I_\ell} &=(1,0)\otimes(1,1)\otimes(0)
	 =\0_0^3,\\
	\g_j|_{I_\ell}=\g_j|_{I_\ell^c}&=(1,1)\otimes(1,0)\otimes(1)
	= \q_{f_\ell(j)}=\q_2.
	\end{align*}
\end{exmp}
\begin{lemma}\label{lem:: some integers i and j}
	Given $i,j\in[0,N)$ with $i<j$ and $N=2^n$, let $n_0\in[0,n)$ be the largest index such that $b_{n_0}(i)=0$ and $b_{n_0}(j)=1$. Then one of the two statements holds:
	\begin{enumerate}
		\item There exists $\ell\in[0,n)\setminus\{n_0\}$ such that $b_\ell(i)\le b_\ell(j)$ and $f_\ell(i)< f_\ell(j)$.
		\item $i=N/2-1$ and $j = N/2$.
	\end{enumerate}
\end{lemma}
\begin{proof}[Proof of Lemma~\ref{lem:: some integers i and j}]
If $n_0<n-1$, we can take $\ell=n-1$ and the first property on the lemma holds. Indeed, $b_\ell(i)=b_\ell(j)$ and 
$$
f_\ell(i)=i-b_{n-1}(i)2^{n-1}<j-b_{n-1}(j)2^{n-1}=f_\ell(j).
$$

Now assume $n_0=n-1$. If the second statement holds, we are done. Otherwise, there must exist $\ell < n - 1$ for
which $b_\ell(i) \le b_\ell(j)$. In this case, we have $f_\ell(i) < 2^{n-2} \le f_\ell(j)$.
\end{proof}	
Let us illustrate how the above lemma works.
\begin{exmp}
	Let $N=4$. We provide Table~\ref{tab::example} showing the integer $n_0$ from Lemma~\ref{lem:: some integers i and j} and a possible choice of $\ell$ for a given pair $(i,j)$. For the case $(i,j)=(1,2)$, we have the second property in Lemma~\ref{lem:: some integers i and j}.
	\begin{table}
		\centering
	\caption{Example demonstrating Lemma~\ref{lem:: some integers i and j}}
	\begin{tabular}{|c|c|c|c|c|c|}
		\hline 
		$(i,j)$& $(0,1)$ & $(0,2)$ & $(0,3)$ & $(1,3)$ & $(2,3)$ \\ 
		\hline 
		base two& $(00,01)$ & $(00,10)$ & $(00,11)$ & $(01,11)$ & $(10,11)$\\
		\hline 
		$\ell$ & $1$ & $0$ & $0$ & $0$ & $1$ \\ 
		\hline 
		$n_0$ & $0$ & $1$ & $1$ & $1$ & $0$ \\ 
		\hline 
	\end{tabular} 
\label{tab::example}
\end{table}
\end{exmp}
\begin{lemma}\label{lem:: some integers j and k}
	Given $j,k\in[0,N)$ with $j<k$, let $\ell$ be an arbitrary integer from $[0,n)$ such that $b_\ell(j)\ge b_\ell(k)$. Then $f_\ell(j)<f_\ell(k)$.
\end{lemma}
\begin{proof}[Proof of Lemma~\ref{lem:: some integers j and k}] Since $j<k$, for the largest integer $n_0$ so that $b_{n_0}(j)\neq b_{n_0}(k)$, we have $b_{n_0}(j)=0<1=b_{n_0}(k)$. Thus, necessarily $\ell \neq n_0$. Let $n_0'$ be the largest integer such that $b_{n_0'}(f_\ell(j)) \neq b_{n_0'}(f_\ell(k))$. 
		 If $\ell < n_0$ then $n_0' = n_0 - 1$, and if $\ell > n_0$ then $n_0' = n_0$. In either case, $b_{n_0'}(f_\ell(j)) = b_{n_0}(j) = 0 < 1 = b_{n_0}(k) = b_{n_0'}(f_\ell(k))$ and, thus, $f_\ell(j) < f_\ell(k)$.
\end{proof}
	\begin{proof}[Proof of Theorem~\ref{theorem on minimal distance}]
		We shall prove the inequality $\w(\x)\ge \w(\g_i+\g_j)$ for any $\x\in C^{(n)}(i,j)$ by induction on $n$, where $N=2^n$. Showing~\eqref{eq::weight of sum} is postponed to the end of the proof.
	The base case $n=1$ is obviously true as $i, j \in [0, 2)$ and condition $i < j$ can be satisfied only for $i = 0, j = 1$. Thus, the set $C^{(1)}(i, j)$ contains only the single element $\g_0 + \g_1$, and, hence, its minimal weight equals $\w(\g_0 + \g_1)$. Assume that the statement holds for $n-1$.  We prove that it is true for $n$. 
		
		Set $N = 2^n$ and let $i, j \in [0, N)$ such that $i < j$. Any binary vector $\x\in C^{(n)}(i,j)$ can be determined by coefficients $(\alpha_k)_{k=j+1}^{N-1} \in\{0,1\}^{N-j-1}$ in the following way
	$$
		\x=\g_i + \g_j+\sum_{k\in[j+1, N) }\alpha_k \g_k.		
	$$
	By Lemma~\ref{lem:: some integers i and j}, we can have two possibilities. If  $i = N/2-1$, then, by the structure of $G_N$ and since $j > i$, $\w(\g_i + \g_j) = \w(\g_i) = N/2$.   As in~\eqref{basic lemma 1}-\eqref{basic lemma 2}, $\w(\x) \ge \w(\g_i) = N/2$. Thus, in this case the statement is true. Otherwise,
	$i \neq N/2-1$ and the first statement of Lemma~\ref{lem:: some integers i and j} must hold. Namely, there exists $\ell\in[0,n)$ such that $b_\ell(i)\le b_\ell(j)$ and $f_\ell(i)< f_\ell(j)$. Define
	\begin{align}
	\x_0&:=\g_i\mathbbm{1}_0(b_\ell(i)) +\g_j\mathbbm{1}_0(b_\ell(j))+\sum_{k\in[j+1, N)}\alpha_k \g_k \mathbbm{1}_0(b_\ell(k)),\nonumber\\
	 \x_1&:=\g_i\mathbbm{1}_1(b_\ell(i)) +\g_j\mathbbm{1}_1(b_\ell(j))+\sum_{k\in[j+1, N)}\alpha_k \g_k  \mathbbm{1}_1(b_\ell(k)), \label{eq::formula for x0 and x1}
	\end{align}
	where the indicator function $\mathbbm{1}_c(x)=1$ if $x=c$ and $\mathbbm{1}_c(x)=0$ if $x\neq c$. It is readily seen that $\x = \x_0+\x_1$. From~\eqref{eq::important 1}-\eqref{eq::important 2} we note that
$$
	\x_0|_{I_\ell}=\0_0^{N/2-1},\quad \x_1|_{I_\ell}=\x_1|_{I_\ell^c}.
$$
	From this it follows
	 	\begin{align}
	\w(\x)&=\w(\x|_{I_\ell})+\w(\x|_{I^c_\ell})\nonumber\\
	&=\w(\x_1|_{I_\ell})+\w(\x|_{I_\ell^c})\nonumber \\&=\w(\x_1|_{I_\ell})+\w(\x_0|_{I^c_\ell}+\x_1|_{I_\ell}).\label{eq::weight of x}
	 \end{align}
	 Recall that $b_\ell(i)\le b_\ell(j)$ and $f_\ell(i)<f_\ell(j)$. We distinguish the further analysis between two cases: $b_\ell(j)=0$ and $b_\ell(j)=1$.
	
	\textbf{Case 1:} $b_\ell(j)=0$. It follows that $b_\ell(i)=0$.
	By applying \eqref{eq::important 1} to $\x_0$ from~\eqref{eq::formula for x0 and x1}, we get
	$$
	\x_0|_{I_\ell^c}=\q_{f_\ell(i)}+\q_{f_\ell(j)}+\sum_{k\in[j+1, N)}\alpha_k \q_{f_\ell(k)}  \mathbbm{1}_0(b_\ell(k)).
	$$
	By Lemma~\ref{lem:: some integers j and k},  $f_\ell(j)<f_\ell(k)$ for $k$ such that $k>j$ and $b_\ell(k)=b_\ell(j)=0$, and, thus, $\x_0|_{I_\ell^c} \in C^{(n-1)}(f_\ell(i), f_\ell(j))$. Therefore, we are able to apply the inductive assumption for $\x_0|_{I_\ell^c}$ and obtain that
	\begin{align}
	\w(\x_0|_{I_\ell^c})&\ge \w(\q_{f_\ell(i)}+\q_{f_\ell(j)})\nonumber\\
	&=\w\left((\g_{i}+\g_{j})|_{I_\ell^c}\right)\nonumber\\
	&=\w\left(\g_{i}+\g_{j}\right),\label{eq::short x0}
	\end{align}
	where the first equality follows from~\eqref{eq::important 1} and the right-most one is implied by~\eqref{eq::important 2} and $b_\ell(i)=b_\ell(j)=0$. Finally, by~\eqref{eq::weight of x} we conclude with
\begin{align*}
	\w(\x)&=\w(\x_1|_{I_\ell})+\w(\x_0|_{I^c_\ell}+\x_1|_{I_\ell}) \\
	&\ge \w(\x_0|_{I_\ell^c})\\&\ge  \w(\g_{i}+\g_{j}),
\end{align*}
where the first inequality is an obvious observation that the weight of $(\x+\y,\y)$ is at least $\w(\x)$ (see similar arguments in~\eqref{basic lemma 1}-\eqref{basic lemma 2}), and the right-most one is implied by~\eqref{eq::short x0}.

\textbf{Case 2:} $b_{\ell}(j) = 1$.	By applying \eqref{eq::important 1} to $\x$, we derive
$$
\x|_{I_\ell^c}=\q_{f_\ell(i)}+\q_{f_\ell(j)}+\sum_{k\in[j+1, N)}\alpha_k \q_{f_\ell(k)}.
$$
	 By Lemma~\ref{lem:: some integers j and k} and since $b_\ell(j)=1\ge b_{\ell}(k)$ for any $k$, we obtain that $f_\ell(j)<f_\ell(k)$ for $k>j$ and, thus,  $\x|_{I_\ell^c}\in C^{(n-1)}(f_\ell(i), f_\ell(j))$. Then we apply the inductive assumption for  $\x|_{I_\ell^c}$ and derive that
\begin{equation}\label{eq::short x0 and x1}
\w(\x|_{I_\ell^c})\ge \w(\q_{f_\ell(i)}+\q_{f_\ell(j)}).
\end{equation}
Now we take a look on the weight of $\x_1|_{I_\ell}$, where $\x_1$ is defined in~\eqref{eq::formula for x0 and x1}. By~\eqref{eq::important 2} we have 
\begin{align*}
\x_1|_{I_\ell}&=\q_{f_\ell(i)} \mathbbm{1}_1(b_\ell(i))+\q_{f_\ell(j)}\\ &+ \sum_{k\in[j+1, N)}\alpha_k \q_{f_\ell(k)} \mathbbm{1}_1(b_\ell(k)).
\end{align*}
For $b_\ell(i)=0$, we obtain $\x_1|_{I_\ell}\in C^{(n-1)}(\0_0^{f_\ell(j)-1},1)$ and by Theorem~\ref{th::first component}, we have 
\begin{equation}\label{eq::second summand with 0}
\w(\x_1|_{I_\ell})\ge\w(\q_{f_\ell(j)}) = \w(\q_{f_\ell(i)}\mathbbm{1}_1(b_\ell(i))+\q_{f_\ell(j)}).
\end{equation}
For $b_\ell(i)=1$, we get $\x_1|_{I_\ell}\in C^{(n-1)}(f_\ell(i), f_\ell(j))$ and
\begin{align}
\w(\x_1|_{I_\ell})&\ge \w(\q_{f_\ell(i)}+\q_{f_\ell(j)}) \nonumber\\&=  \w(\q_{f_\ell(i)}\mathbbm{1}_1(b_\ell(i))+\q_{f_\ell(j)}).\label{eq::second summand with 1}
\end{align}
Finally by~\eqref{eq::weight of x} we conclude with
\begin{align*}
\w(\x)&=\w(\x_1|_{I_\ell})+\w(\x|_{I_\ell^c})
\\&\ge \w(\q_{f_\ell(i)}+\q_{f_\ell(j)}) + \w\left (\q_{f_\ell(i)}\mathbbm{1}_1(b_\ell(i))+\q_{f_\ell(j)}\right) \\&= \w(\g_{i}+\g_{j}),
\end{align*}
where the inequality is due to~\eqref{eq::short x0 and x1}-\eqref{eq::second summand with 1}, and the right-most equality is implied by~\eqref{eq::important 1}-\eqref{eq::important 2}. This completes the proof of cases 1 and 2, and we are now ready to tackle~\eqref{eq::weight of sum}.

For any $i$ and $j$, we have that $\w(\g_i+\g_j) = \w(\g_i) + \w(\g_j) - 2 T_{i,j}$, where $T_{i,j}$ denotes the number of positions at which  $\g_i$ and $\g_j$ are both one, i.e. $T_{i,j}:=|\s(\g_i)\cap \s(\g_j)|$.  Thus, to verify~\eqref{eq::weight of sum}, it suffices to check that $T_{i,j}=2^{t_{i,j}}$  for any $i$ and $j$. We shall check this equality by induction on $n$.
 
For the base case $n=1$, we have either $(i,j)\in\{(0,0), \ (0,1),\ (1,0)\}$ with $t_{i,j}=0$ and $T_{i,j}=1$, or $(i,j)=(1,1)$ with $t_{i,j}=1$ and $T_{i,j}=2$. Assume the statement holds for $n-1$. Let us show that it is true for $n$. By \eqref{eq::representation of gi}, we have that
$$
\g_i=\gg_{b_{n-1}(i)}\otimes \q_{f_{n-1}(i)},\quad \g_j=\gg_{b_{n-1}(j)}\otimes \q_{f_{n-1}(j)},
$$
where $\gg_0=(1,0)$ and $\gg_1=(1,1)$. This implies  the property
$|\s(\g_i)\cap \s(\g_j)|=
2|\s(\q_{f_{n-1}(i)})\cap \s(\q_{f_{n-1}(j)})|$ if  $b_{n-1}(i)=b_{n-1}(j) = 1$,  and 
$|\s(\g_i)\cap \s(\g_j)|=|\s(\q_{f_{n-1}(i)})\cap \s(\q_{f_{n-1}(j)})|$ otherwise.
By the inductive assumption, we obtain that 
\begin{align*}
T_{i,j} = \ &|\s(\g_i)\cap \s(\g_j)|\\
=\ &2^{b_{n-1}(i) b_{n-1}(j)}|\s(\q_{f_{n-1}(i)})\cap \s(\q_{f_{n-1}(j)})|\\
=\ & 2^{b_{n-1}(i) b_{n-1}(j)} 2^{t_{f_{n-1}(i),f_{n-1}(j)}} \\
=\ & 2^{\sum_{k=0}^{n-1}b_{k}(i) b_{k}(j)}\\
=\ & 2^{t_{i,j}}.
\end{align*}
This completes the proof.
	\end{proof}

	\section{Conclusion}\label{conclusion}
	In this paper, we discuss the exact weight distribution of the coset associated with each synthetic channel $W_N^{(i)}$. Also, we find the minimal distance between cosets associated with paths that differ in two positions for successive cancellation list decoding. This study represents initial steps towards understanding the performance of polar codes under successive cancellation list decoding.
	
	The approximate union bound~\eqref{union bound} takes into account only weight distributions of the coset $C^{(n)}(\0_0^{i-1}, 1)$. The drawback of this approach is evident: for low and medium signal-to-noise ratio, the estimate cannot be tight. Based on Remark~\ref{rem::computing wd for zero coset}, the weight distributions of any zero coset $C^{(n)}(\0_0^{i-1}, 0)$  can be calculated. However, we do not know how to use it in order to get a more accurate bound of the error probability.
	It is still unknown how to calculate efficiently the minimal weight word (and the number of minimal weight words) of a set 
	$C^{(n)}(\u_0^{i-1})$ for an arbitrary $\u_0^{i-1}$. We believe that such an analysis can be helpful for constructing polar codes under successive cancellation list decoding and estimating the distance spectrum of polar codes. 
	
	\bibliographystyle{IEEEtran}
	\bibliography{root}
	
\end{document}